\documentclass{amsart}
\usepackage{amsmath}
\usepackage{amssymb}
\usepackage[all]{xy}
\usepackage{comment}
\usepackage{enumitem}
\usepackage{color}
\usepackage{setspace}

\theoremstyle{plain}
\newtheorem{theorem}{Theorem}[section]
\newtheorem{thm*}{Theorem}[section]
\newtheorem{corollary}[theorem]{Corollary}

\newtheorem{lemma}[theorem]{Lemma}

\newtheorem{lemma*}{Lemma}

\theoremstyle{definition}
\newtheorem{definition}[theorem]{Definition}
\newtheorem{remark}[theorem]{Remark}

\newtheorem*{remark*}{Remark}

\newtheorem{example}[theorem]{Example}

\newtheorem{question*}{Question}
\DeclareMathOperator{\esssup}{ess\,sup}

\numberwithin{equation}{theorem}

\SelectTips{eu}{}

\date\today

\begin{document}

 \title[On the Coherent Risk Measure Representations in the  Discrete Probability Spaces]{On the Coherent Risk Measure Representations in the  Discrete Probability Spaces}

 \author[Kerem U\u{g}urlu]
{Kerem U\u{g}urlu}

\address {Department of Mathematics, University of Southern California,
Los Angeles, CA}
\email{kugurlu@usc.edu}


\begin{abstract}
We give a complete characterization of both comonotone and not comonotone coherent risk measures in the discrete finite probability space,  where each outcome is equally likely. To the best of our knowledge, this is the first work that \textit{ characterizes and distinguishes} comonotone and not comonotone coherent risk measures via AVaR representation in the discrete finite probability space of equally likely atoms. The characterization gives a more efficient and exact way of representing the law invariant coherent risk measures in this probability space, which is crucial in applications and simulations.
\end{abstract}

\maketitle

\section{Introduction}
In the seminal paper of Artzner et al. \cite{ADEH}, the coherent risk measures are introduced and their properties are axiomatized. Risk measures have gained a lot of interest both in theory and applications since then. This paper addresses representations of the coherent risk measures in the discrete probability space $(\Omega,2^{|\Omega|},\mathbb{P})$, where $|\Omega| = n$ and $\mathbb{P}(\omega(i))=1/n$ for all $i=1,2,...,n$. The purpose of this paper is to give a complete characterization of both comonotone and not neccessarily comonotone coherent risk measures in this probability space. We appeal to the natural risk statistics formulation of \cite{HKP} and to the concept of functional coherence introduced independently in \cite{NG} and \cite{PS} and represent \textit{any} coherent risk measure in this discrete probability space. The closest works in this direction are \cite{BB} and \cite{NG}. In both of these works, the characterization of \textit{comonotone} coherent risk measures are given via AVaR as building blocks. In \cite{NG}, the not comonotone case is given via a supremum of AVaR sums formulation, whereas in \cite{BB} there is no referring to not comonotone case. To the best of our knowledge, this is the first work that characterizes \textit{and} distinguishes comonotone and not comonotone coherent risk measures via a \textit{simplified} AVaR representation in this probability space, which is crucial in the applications and simulations.

The rest of the paper is as follows. In Section 2, we give the theoretical background along with the necessary definitons. In Section 3, we show first that any coherent risk measure is SSD preserving  in the discrete  finite probabilty space, where each outcome is equally likely. Then, we give our two theorems which give the characterizations of comonotone and not comonotone coherent risk measures in this probability space. In the rest of the paper, we prove these two results.

\section{Preliminaries and Theoretical Background}\label{Review}

Let $(\Omega, \mathcal{F},\mathbb{P})$ be the atomless standard probability space. Hence without loss of generality we can take $\Omega$ to be the unit interval $[0,1]$, $\mathbb{P}$ to be the Lebesgue measure with $\mathcal{F}$ to be the Borel sigma algebra. Let $(\Omega, \mathcal{G},\mathbb{P})$ be the probability space, where $|\Omega| = n $, $\mathcal{G} = 2^{\Omega}$ and $\mathbb{P}$ is a probability measure that satisfies $\mathbb{P}(\omega_i) =1/n$ for all $i \in \{1,2,...,n\}$. We call this probability space the uniform discrete probability space. A random variable (r.v.) $X$ is a measurable function from $\Omega$ to $\mathbb{R}$. The cumulative distribution function of a r.v. is defined by $F_X(x) = \mathbb{P}( X \leq x )$. The $p$-quantile of a r.v. $X$ is denoted by $\mathrm{VaR}_p(X):=\inf \{x:P(X \leq x)\geq p \}$, which is left-continuous and lower semi-continuous.
\begin{definition}
Given two r.v.'s $X$ and $Y$, we say $X$ second-order stochastically dominates (SSD) $Y$ and write $X \succeq Y$, if
\begin{equation}
\int_{-\infty}^{t} F_X(s) ds \leq  \int_{-\infty}^{t} F_Y(s) ds, \qquad \forall t \in \mathbb{R}.
\end{equation}
\end{definition}
\begin{definition}
The coherent risk measure $\rho$ is a function that is mapping $\mathbb{R}$-valued r.v.'s into the real numbers $\mathbb{R}$ or to $+\infty$, which satisfies the following axioms
\begin{itemize}
\item (monotonicity): $\rho(Y_1) \leq \rho(Y_2)$ whenever $Y_1 \leq Y_2$ almost surely.
\item (positive homogeneity): $\rho(\lambda Y ) = \lambda \rho(Y)$ whenever $\lambda > 0$.
\item (convexity) $\rho( ( 1 - \lambda )Y_0 + \lambda Y_1 ) \leq ( 1 - \lambda )\rho(Y_0) + \lambda \rho (Y_1)$ for $0 \leq \lambda \leq 1 $.
\item (translation invariance) $\rho(Y+c) = \rho(Y) + c$ if $c \in \mathbb{R}$. 
\end{itemize}
\end{definition}
\begin{definition}
A coherent risk measure $\rho$ is called law invariant if two r.v.'s $X$ and $Y$ on a probability space having the same distribution implies that $\rho(X) = \rho(Y)$.
\end{definition}
An important coherent risk measure that we will use throughout the paper is the Average-Value-at-Risk denoted by $\mathrm{AVaR}_{\alpha}(Y)$
\begin{equation}
\label{rep1}
\mathrm{AVaR}_{\alpha}(Y) := \frac{1}{1 - \alpha} \int_{\alpha}^1 \mathrm{VaR}_u(Y) du
\end{equation}
An alternative representation to \eqref{rep1} for $\mathrm{AVaR}_{\alpha}(Y)$ is given in \cite{RU} with the following form 
\begin{equation}
\label{rep2}
\mathrm{AVaR}_{\alpha}(Y) = \min_{s \in \mathbb{R}} \left \{ s + \frac{1}{1-\alpha} \mathbb{E}[ (X - s)^{+} ] \right\}
\end{equation}
where the minimum in \eqref{rep2} is attained at $\mathrm{VaR}_{\alpha}(Y)$.
\begin{remark}
\label{cont} We note from \eqref{rep2} that $\alpha \rightarrow \mathrm{AVaR}_{\alpha}(Y)$ is a continuous function with respect to variable $\alpha$ on the interval $[0,1)$. Note also that $\mathrm{AVaR}_0(Y) = \mathbb{E}[Y]$ and $\lim_{\alpha \rightarrow \infty} \mathrm{AvaR}_{\alpha} (Y) = \esssup [Y]$. 
\end{remark}
Moreover, it is shown in \cite{RS} via the Fenchel-Moreau theorem (see e.g.\cite{RW}) that we have the following equivalent representation for $\mathrm{AVaR}_{\alpha}(Y)$, and $Y \in L^p(\Omega, \mathcal{H}, \mathbb{P})$, with $p \geq 1$
\begin{equation}
\label{avar_rep}
\mathrm{AVaR}_{\alpha}(Y) = \sup_{ \mu \in \mathcal{C}} \langle \mu, Y \rangle 
\end{equation}
where $\mathcal{C}$ is the set of probability densities with absolutely continuous probability densities $h \in L^{q}(\Omega, \mathcal{H}, \mathbb{P})$ with respect to underlying reference probability measure $\mathbb{P}$ satisfying
\begin{equation}
\mathcal{C} = \left\{ h \in L^{q} : 0  \leq h \leq \frac{1}{1-\alpha}, \int_{\Omega} h d\mathbb{P} = 1  \right\}.
\end{equation}
Here $L^{q}(\Omega, \mathcal{H}, \mathbb{P})$ is the dual of $L^{p}(\Omega, \mathcal{H}, \mathbb{P})$. It is also the case, that supremum in \eqref{avar_rep} is attained, whenever $1 \leq p < \infty $.
\begin{remark}
Note that, in the discrete uniform case with $P(\omega(i) = \frac{1}{n})$, we immediately get that the absolutely continuous probability density functions $h(\omega)$ are of the form 
\begin{equation}
h(\omega_i) \leq \min \{\frac{1}{n-i},1 \} ,\mbox{ for all } 1 \leq i \leq n.
\end{equation}
\end{remark}
\begin{remark}
We also remark that, in fact, more general is true. Due to Fenchel-Moreau theorem, \textit{any} law invariant coherent risk measure in $L^p(\Omega, \mathcal{H}, \mathbb{P}), 1 \leq p \leq \infty $ has the representation
\begin{equation}
\label{coh_rep}
\rho(X) = \sup_{ \nu \in \mathcal{D}} \langle \mu, X \rangle 
\end{equation}
where $\mathcal{D}$ is a convex set of absolutely continuous probability densities of $\nu$ with respect to reference probability  measure $\mathbb{P}$ in the dual of $L^p$ (see \cite{ADEH}). 
\end{remark}
We will need the following dependence property of random variables and coherent risk measures, correspondingly.
\begin{definition}
A pair of r.v.'s $X$ and $Y$ is said to be comonotone, if the following condition holds.
\begin{equation}
(X(\omega_1) - X(\omega_2))(Y(\omega_1) - Y(\omega_2)) \geq 0 \text{ a.s. }
\end{equation}
Similarly, a coherent risk measure $\rho$ is said to be comonotone additive, if for every pair of comonotone r.v.'s $X$ and $Y$
\begin{equation}
\rho(X+Y) = \rho(X) + \rho(Y)
\end{equation}
holds.
\end{definition}
\begin{remark}
We know by \cite{P} that $\mathrm{AVaR}_{\alpha}(X)$ is comonotone additive whenever $ \alpha < 1$ and, by considering the continuity of $\mathrm{AVaR}_{\alpha}(X)$, whenever $\mathrm{AVaR}_{\alpha}(X) \neq \esssup[X]$. However, below we provide a simple example that the coherent risk measure $\esssup[X]$ is not comonotone additive.
\end{remark}
\begin{example}
Let $\Omega$ be the discrete uniform probability space of four atoms $\omega_1, \omega_2, \omega_3, \omega_4$. Let $X(\omega_1) = Y(\omega_1) = 0$, $X(\omega_2) = Y(\omega_2) = 1$, $X(\omega_3) = 0.8$, $Y(\omega_3) = -1$ and $X(\omega_4) = 3$ and $Y(\omega_4) = 0.5$.
Note that $X$ and $Y$ are comonotone but $\esssup[X+Y] < \esssup[X] + \esssup[Y]$.
\end{example}
In his seminal work Kusuoka \cite{K} showed the following characterization of law invariant coherent risk measure on the atomless probability space $(\Omega, L^{\infty}, \mathbb{P})$, which later extended to the atomless $L^p, p \geq 1$ case (see \cite{DRS} and \cite{JST}). 
\begin{theorem} \cite{K}
A mapping $\rho : L^{p} \rightarrow \mathbb{R} \cup \{ \infty \}, p \geq 1$  on an atomless probability space $(\Omega, \mathcal{F},\mathbb{P})$ is a law invariant coherent risk measure if and only if it admits the following representation 
\begin{equation}
\label{kusuoka}
\rho (X) = \sup_{\mu \in \mathcal{M}} \int_{[0,1]} \mathrm{AVaR}_{t}(X)d\mu_t
\end{equation}
for any r.v. $X$, where $\mathcal{M}$ is a probability measure on $[0,1]$. 
If, in addition $\rho$ is comonotone additive, then supremum is attained in \eqref{kusuoka} for a probability measure $\mu^*$ on $[0,1]$ such that 
\begin{equation}
\label{kusuoka_com}
\rho (X) = \int_{[0,1]} \mathrm{AVaR}_{t}(X)d\mu^*_t.
\end{equation}
\end{theorem}
We proceed with the following definition.
\begin{definition}
A coherent risk measure $\rho$ is said to preserve SSD if $X \succeq Y$ implies that $\rho(X) \geq \rho(Y)$
\end{definition}
The following result shows the strong dependence of SSD-preservation and coherent risk measures, (see, \cite{FS}, Theorem 2.58 and Remark 4.38.)
\begin{theorem}
For $X,Y \in L^{\infty}$ the following conditions are equivalent:
\begin{itemize}
\item $X \preceq Y$,
\item $\mathbb{E}[U(X)] \leq \mathbb{E}[U(Y)]$ for all nondecreasing concave functions $U$ on $\mathbb{R}$,
\item $\mathrm{AVaR}_{\alpha}(Y) \leq \mathrm{AVaR}_{\alpha}(X)$ for all $\alpha \in [0,1]$.
\end{itemize}
\end{theorem}
Leitner \cite{L} showed that admitting Kusuoka representation \eqref{kusuoka} and preserving SSD are exactly the same properties of the coherent risk measure $\rho$ in the general probability space.
\begin{theorem} \cite{L}\label{leitner}
In a not necessarily atomless probability space $(\Omega, \mathcal{F},\mathbb{P})$, a coherent risk measure $\rho$ admits Kusuoka representation \eqref{kusuoka} iff $\rho$ is SSD preserving.
\end{theorem}
Next, we give the concept that is introduced in \cite{NG} and \cite{PS} independently.
\begin{definition}
Given a not necessarily atomless probability space $(\Omega, \mathcal{G}, \mathbb{P})$, a law invariant mapping $\rho(X)$ on $(\Omega, \mathcal{G}, \mathbb{P})$ is called a functionally coherent risk measure, if there exists a law invariant coherent risk measure $\varrho(X)$ defined on the standard atomless probability space $(\Omega, \mathcal{F}, \mathbb{P})$ such that $\rho(X) = \varrho \lvert_{(\Omega, \mathcal{F}, \mathbb{P})} (X) $. In addition to the above property, if $\rho(X)$ is also comonotone additive, then we say that $\rho(X)$ is a functionally coherent and comonotone additive risk measure.
\end{definition}
We conclude this section with the definition introduced in \cite{HKP} and \cite{AFS}  related to the underlying probability distribution of the random variable $X$.
\begin{definition}
If the underlying probability distribution is discrete and finite with $|\Omega| = n$, then a coherent risk measure $\rho$ is said to be permutation invariant if $\rho(X_{\pi}) = \rho(X)$ for every permutation $\pi \in S_n$, where $S_n$ is the set of all permutations of $\{1,2,...,n\}$ and $X_{\pi}$ denotes the permuted vector, i.e. $X_{\pi} = (x_{\pi(1)},...,x_{\pi(n)})$. A coherent risk measure $\rho(X): \mathbb{R}^n \rightarrow \mathbb{R}$ that is also permutation invariant is called a natural risk statistic.
\end{definition}
\section{Main Results}\label{Rationality}
In this section, we give our main results along with the proofs. 
The first result shows that in the finite discrete uniform probability space, any law invariant coherent risk measure preserves the SSD property.
\begin{theorem}
\label{main_thm1}
A law invariant coherent risk measure $\rho$ on the uniform discrete probability space is SSD preserving.
\end{theorem}
\begin{proof}
Given two r.v.'s $X$ and $Y$, denote them with 
\begin{align}
X &= (x_1,x_2,...,x_n) \nonumber \\
Y &= (y_1,y_2,...,y_n)
\end{align}
Since the probabilities of all elementary events are equal, the SSD relation coincides with the concept of weak majorization (see \cite{MO}).
\begin{equation}
[ X \preceq Y ] \iff [ \sum_{k=1}^n x_{[k]} \leq \sum_{k=1}^n y_{[k]} ], 
\end{equation}
where $x_{[k]}$ denotes the kth smallest component of X. It follows from the theorem by Hardy, Littlewood and Polya (see \cite{HLP} and Proposition D.2.b in \cite{MO}) that weak majorization is equivalent to the existence of a doubly stochastic matrix $A$ such that 
\begin{equation}
X \leq AY.
\end{equation}
According to Birkhoff's theorem (see \cite{B}), the doubly stochastic matrix $A$ is a convex combination of permutation matrices, i.e. there are permutation matrices $B_j$ and weights $\alpha_j$, $j=1,..,n$ with $\alpha_j \geq 0$ and $\sum_{j=1}^n \alpha_j = 1$ such that 
\begin{equation}
A = \sum_{j=1}^n \alpha_j B_j
\end{equation}
Define random variables $Z_j : \Omega \rightarrow \mathbb{R}$ such that $Z_j(w_i)$ is the ith component of $B_j y$. Then $Z_j, j = 1,...,n$ all have the same distribution as $Y$ and it follows from above that %
\begin{equation}
X \leq \sum_{j=1}^n \alpha_j Z_j
\end{equation}
Monotonicity,convexity and law invariance imply 
\begin{align}
\rho(X) &\leq \rho(  \sum_{j=1}^n \alpha_j Z_j ) \leq \sum_{j=1}^n \alpha_j \rho(Z_j)
\nonumber\\
&= \sum_{j=1}^n \alpha_j \rho(Y) = \rho(Y)
\end{align}
so we complete the proof.
\end{proof}
Theorem 3.1 will lead to the simplified characterizations of the coherent risk measures in the finite discrete uniform probability space. We state these two representations below.
\begin{theorem}
\label{main_thm2}
Let $\mathbb{P}:= \{ W \in \mathbb{R}^n | \sum_{j=1}^n w_j = 1,\quad w_j \geq 0, \quad j = 1,...,n \}$ and $\mathbb{D} := \{ X \in \mathbb{R}^n | x_1 \leq x_2 \leq ... \leq x_n \}$. Any coherent and comonotone additive risk measure $\rho$ in the uniform discrete probability space has the following form 
\begin{equation}
\rho(X) = \sum_{i=0}^{n-1} \mu_i \mathrm{AVaR}_{i/n} (X), \qquad \forall X \in \mathbb{R}^n,
\end{equation}
where
\begin{align}
0 &\leq \mu_i \leq 1, \text{ for all } i=0,1,...,n-1  \nonumber \\
\sum_{i=0}^{n-1} \mu_i &= 1 \nonumber \\
\mathrm{AVaR}_{\frac{i}{n}}(X) &= \frac{1}{n-i}\left( X_{[i+1]} +...+ X_{[n]} \right).
\\\nonumber
\end{align}
i.e. $\mathrm{AVaR_{\frac{i}{n}}(X)}$ satisfies
\begin{align}
\mathrm{AVaR}_{\frac{i}{n}}(X) &= \max_{W \in \mathbb{P}\cap \mathbb{D} } \langle W,X_{os} \rangle  \\\nonumber
&= w_1 X_{[1]} + w_2 X_{[2]} + ... + w_n X_{[n]}\nonumber\\
0 &\leq w_j \leq \frac{1}{n - i}, \text{ for all } j=1,2,...,n \nonumber\\
0 &\leq w_1 \leq w_2 \leq ... \leq w_n \leq 1,  \nonumber\\
\sum_{j=1}^n w_j &= 1 \nonumber
\end{align}
\end{theorem}
Similarly, for the coherent risk measure $\rho$, which is not comonotone additive, we have the following characterization.
\begin{theorem}
\label{main_thm3}
Let $\mathbb{P}$ and $\mathbb{D}$ be as above. Any coherent but not comonotone additive risk measure $\rho$ in the uniform discrete probability space has the following form 
\begin{equation}
\rho(X) = \sum_{i=0}^{n-1} \mu_i \mathrm{AVaR}_{i/n} (X) + \mu_n \mathrm{AVaR}_{1} (X), \qquad \forall X \in \mathbb{R}^n,
\end{equation}
\begin{align}
0 &\leq w_j \leq \frac{1}{n - i}, \forall j=1,2,...,n \nonumber\\ 
0 &\leq w_1 \leq w_2 \leq ... \leq w_n \leq 1,  \nonumber\\
\sum_{j=1}^n w_j &= 1  \nonumber \\
0 &\leq \mu_i \leq 1 \nonumber \text{ for all } i = 0,1,...,n\\
\mu_n &> 0 \nonumber \\
\sum_{i=0}^{n} \mu_i &= 1 \nonumber 
\end{align}
\end{theorem}
In the rest of the paper, we give the proofs of the Theorem 3.2 and Theorem 3.3 and also derive several corollaries on the way. We borrow a result from \cite{NG}, which gives the equivalence of functional coherence and Kusuoka representation on a general probability space.
\begin{theorem}\cite{NG}
Consider a (not necessarily atomless) probability space $(\Omega,\mathcal{G},\mathbb{P})$, and a value $p \in [1,\infty]$.
\begin{itemize}
\item A mapping $\rho$ is a functionally coherent risk measure if and only if it has the Kusuoka representation of the form \eqref{kusuoka} for some family of probability measures $\mathcal{M}\subset \mathbb{P}$.
\item A mapping $\rho$ is a functionally coherent and comonotone additive risk measure if and only if it has the representation \eqref{kusuoka_com}.
\end{itemize}
\end{theorem}
Based on Theorem 3.4, we get the following characterization of law invariant coherent risk measures immediately.
\begin{theorem}
\label{SSD_fc}
A law invariant coherent risk measure $\rho$ on the (not necessarily atomless) probability space is SSD preserving if and only if $\rho$ is functionally coherent.
\end{theorem}
\begin{proof}
A law invariant coherent risk measure $\rho$ is SSD preserving iff $\rho$ admits a Kusuoka representation by \cite{L}, and by \cite{NG}, $\rho$ admits the Kusuoka representation iff it is functionally coherent, thus we conclude the proof.
\end{proof}
\begin{corollary} \label{uni_fc} Any law invariant coherent risk measure $\rho(X)$ on the discrete uniform probability space is necessarily functionally coherent.
\end{corollary}
\begin{proof}
By Theorem 3.1, on the discrete uniform probability space, a coherent risk measure $\rho(X)$ is SSD preserving. Hence the result follows from Theorem 3.5. 
\end{proof}
We proceed with the lemma below.
\begin{lemma}
On the discrete probability space with uniform distribution, the coherent risk measure $\mathrm{AVaR}_{\alpha}(X)$ is permutation invariant.
\begin{proof}
Let $W = \{w_1,,w_2,...,w_n\}$ be the vector of nonnegative weights with\\ $\sum_{i=1}^n w_i = 1$, and $X = \{x_1, x_2,...,x_n\}$ be the vector form  of r.v. $X$. Then by $\eqref{avar_rep}$ that 
\begin{align}
\label{avar_rep2}
\mathrm{AVaR}_{\alpha}(X) &= \max_{\{0 \leq w \leq 1 / n(1-\alpha) \}} \langle X,\mu \rangle \nonumber \\
           &= w_1 X_1 + w_2 X_2 +...+ w_n X_n
\end{align}
Hence, we see that when we interchange $x_i$ with $x_j$ where $i \neq j$, by interchanging the $w_i$ and $w_j$ we get the same result. This implies that $\mathrm{AVaR}_{\alpha}(X)$ is permutation invariant. 
\end{proof}
\end{lemma}
Next we will state the result of \cite{AFS} and \cite{HKP} for natural risk statistics.
\begin{theorem} Let $D := \{ X \in \mathbb{R}^n | x_1 \leq x_2 \leq ... \leq x_n \}$ and $\mathbb{P} := \{ X \in \mathbb{R}^n | \sum_{i=1}^n x_i = 1, x_i\geq 0, i=1,...,n \}$ and denote by $X_{os}$ the order statistics of $X$, i.e. $X_{os} := (x_{[1]},x_{[2]},...,x_{[n]})$ for some $\pi \in S_n$ such that $X_{\pi} \in D$.
Suppose the natural risk statistic $\rho$ is subadditive. Then there exists a closed convex set of weights $W \subset \mathbb{P} \cap \mathbb{D}$ such that 
\begin{equation}
\label{nat}
\rho(X) := \max_{W \in \mathcal{W}} \langle W,X_{os} \rangle, \qquad \forall X \in \mathbb{R}^n.
\end{equation}
\end{theorem}
We are now ready to prove Theorem 3.2 and Theorem 3.3.
\begin{proof}
[Proof of Theorem 3.2] Using the definition \eqref{rep1}, by a straightforward calculation we get that 
\begin{equation}
\label{avar_uni}
\mathrm{AVaR}_{i/n} (X) = \frac{1}{n-i}[X_{[i+1]} + ... +X_{[n]}]
\end{equation}
Moreover, for $\frac{i-1}{n} \leq \alpha \leq \frac{i}{n}$, let $\lambda = \frac{n(1-\alpha)-(n-i+1)(i-n\alpha)}{n(1-\alpha)}$. Then, it is easy to verify that $0\leq \lambda \leq 1$ with
\begin{equation}
\mathrm{AVaR}_{\alpha}(X) = \lambda \mathrm{AVaR}_{\frac{i}{n}}(X) + (1-\lambda)\mathrm{AVaR}_{\frac{i-1}{n}}(X)
\end{equation}
Then, by Theorem 3.1, we have, for the comonotone coherent risk measure $\rho$
\begin{align}
\rho(X) &= \int_0^1 \mathrm{AVaR}_p(X)d\mu(p)
\nonumber\\
&= \sum_{i=1}^n \int_{\frac{i-1}{n}}^{\frac{i}{n}} \mathrm{AVaR}_p(X)d\mu(p)
\nonumber\\
&= \sum_{i=1}^n \int_{\frac{i-1}{n}}^{\frac{i}{n}} \left(p\mathrm{AVaR}_{\frac{i-1}{n}}(X) + (1-p)\mathrm{AVaR}_{\frac{i}{n}}(X)\right)d\mu(p)
\nonumber\\
&= \sum_{i=1}^n \mathrm{AVaR}_{\frac{i-1}{n}}(X)\int_{\frac{i-1}{n}}^{\frac{i}{n}} p ~d\mu(p) + \mathrm{AVaR}_{\frac{i}{n}}(X)\int_{\frac{i-1}{n}}^{\frac{i}{n}} (1-p) ~d\mu(p)
\end{align}
We note here that the positive coefficients of $\mathrm{AVaR}_{\frac{i}{n}}$ for $0 \leq i \leq n$ add up to $\int_0^1 d\mu(p) = 1$, since $\mu$ is a probability measure on [0,1]. We denote those coefficients as $\mu_i$. We also know that the last coefficient  $\mu_n = \int_{\frac{n-1}{n}}^{1} (1-p)d\mu(p) = 0$, since $\rho$ is a comonotone coherent risk measure. 
Moreover, since we are in the discrete uniform probability space, by Lemma 3.8. $\mathrm{AVaR}_\alpha(X)$ is permutation invariant, hence a natural risk statistic. It is also subadditive by being a coherent risk measure. Thus, the representation \eqref{avar_uni} necessarily satisfies \eqref{avar_rep2} where 
\begin{align}
0 &\leq w_j \leq \min\left\{\frac{1}{n-i},1\right \} \text{, for all } 1\leq i \leq n 
\nonumber\\
0 &\leq w_1 \leq w_2 \leq ... \leq w_n \leq 1  
\nonumber\\
&\sum_{j=1}^n w_j = 1,
\end{align}
which implies $w_j = \frac{1}{n-i}$ for $i+1 \leq j \leq n$ and $w_j = 0$ for $1 \leq j \leq i$.
\end{proof}
\begin{proof}
[Proof of Theorem 3.3]  Modulo the same arguments as in the proof of Theorem 3.2,  we have that any law invariant coherent risk measure is of the form 
\begin{equation}
\label{son}
\rho(X) = \sup_{\mu \in \mathcal{M}}\sum_{i=0}^{i=n} \mu_{i} \mathrm{AVaR}_{i/n}(X).
\end{equation}
But, since $\mu = \{\mu_0,\mu_1,\mu_2,...,\mu_n\}$ is an element in the unit simplex in $\mathbb{R}^{n+1}$, by Heine-Borel Theorem, representation \eqref{son} attains its maximum for a specific $\mu$. We also know by \cite{P} that $\mathrm{AVaR}_{i/n}(X)$ is comonotone additive for $0 \leq i \leq n-1 $. Hence any combination with $\mu_n = 0$ is necessarily comonotone, and the weight $\mu_n$ is strictly positive. The theorem is proven.
\end{proof}
\begin{remark}
Note that, for a finite uniform probability space with equal probabilities, two random variables $X$ and $Y$ are distributionally equivalent if and only if there is a permutation transferring one into another. Hence, any law invariant coherent risk measure $\rho$ is necessarily permutation invariant, and is a natural risk statistic in the discrete finite uniform probability space due to \eqref{coh_rep}, as well. Thus, we note that the formulations of the coherent risk measure in Theorem 3.2 and Theorem 3.3 satisfy the representation \eqref{nat} in Theorem 3.8. This also reveals that being more \textit{risk averse}, i.e. adding more weight to $X_{[n]}$ causes that the coherent risk measure $\rho$ loses its comonontonicity property.
\end{remark}

\end{document}